\documentclass{article}
\usepackage[text={16cm,21cm},centering]{geometry}
\usepackage[hidelinks]{hyperref}
\usepackage{amsmath}
\usepackage{amsthm}
\usepackage{amsfonts}
\usepackage{amssymb}
\usepackage{amscd}
\usepackage[mathscr]{eucal}
\usepackage{mathtools}
\usepackage{fullpage}
\usepackage{parskip}
\usepackage{color}
\usepackage{float}
\usepackage{graphicx}
\usepackage{natbib}

\newtheorem{proposition}{Proposition}

\newtheorem{corollary}{Corollary}
\begin{document}
\title{\bf On orthogonal transformations of Christoffel equations}
\author{Len Bos%
\footnote{Dipartimento di Informatica, Universit\`a di Verona, Italy;
{\tt leonardpeter.bos@univr.it}}\,,
Michael A. Slawinski%
\footnote{Department of Earth Sciences, 
Memorial University of Newfoundland, Canada; 
{\tt mslawins@mac.com}}\,,
Theodore Stanoev%
\footnote{Department of Earth Sciences, 
Memorial University of Newfoundland, Canada; 
{\tt theodore.stanoev@gmail.com}}\,,
Maurizio Vianello%
\footnote{Dipartimento di Matematica, 
Politecnico di Milano, Italy; 
{\tt maurizio.vianello@polimi.it}}}
\date{}
\maketitle
\paragraph{Abstract:}
The purpose of this paper is to prove the equivalence---under rotations of distinct terms---of different forms of a determinantal equation that appears in the studies of wave propagation in Hookean solids, in the context of the Christoffel equations.
To do so, we prove a general proposition that is not limited to~${\mathbb R}^3$\,, nor is it limited to the elasticity tensor with its index symmetries.
Furthermore, the proposition is valid for orthogonal transformations, not only for rotations.
The sought equivalence is a corollary of that proposition.
\section{Introduction}
The existence and properties of three waves that propagate in a Hookean solid are a consequence of the Christoffel equations \citep[e.g.,][Chapter~9]{Slawinski}, whose solubility condition is
\begin{equation*}
\det\left[
	\sum_{j=1}^3\sum_{\ell=1}^3 c_{ijk\ell}\,p_{j}p_{\ell} -
	\delta_{ik}
\right]=0\,,\qquad i,k=1,2,3\,,
\end{equation*}
which is a cubic polynomial, whose roots are the eikonal equations \citep[e.g.,][Section~7.3]{Slawinski}.
Let us examine the matrix therein,
\begin{equation*}
\left[
	\sum_{j=1}^3 \sum_{\ell=1}^3 c_{ijk\ell}\,p_{j}p_{\ell}
\right]\in {\mathbb R}^{3\times 3}\,,
\end{equation*}
where $c_{ijk\ell}$ is a density-normalized elasticity tensor, whose units are $\rm km^2\,s^{-2}$\,, and $p$ is the wavefront-slowness vector, whose units are $\rm s\,km^{-1}$\,.

Studies of Hookean solids by \citet[equations~(7)--(12)]{IvanovStovas2016} and \citet[equations~(10)--(11)]{IvanovStovas2017} invoke a property that we state as Corollary~\ref{cor:Michael}, which is a consequence of Proposition~\ref{prop:Len}.
\citet{IvanovStovas2016,IvanovStovas2017} verify the equivalence of equations given in Corollary~\ref{cor:Michael}, without a general proof, hence, this paper.

The purpose of this paper is to prove Proposition~\ref{prop:Len} and, hence, Corollary~\ref{cor:Michael}.
In doing so, we gain an insight into a tensor-algebra property that results in this corollary.
The equivalence of the aforementioned equations is not a result of the invariance of a determinant, as suggested by Ivanov (pers. comm.,~2018); it is a consequence of two orthogonal transformations of $c_{ijk\ell}$ and $p_i$ that result in two matrices that are similar to one another.
\section{Proposition and its corollary}
\begin{proposition}
\label{prop:Len}
Consider a tensor,~$c_{ijk\ell}$\,, in ${\mathbb R}^d$\,.
Also, consider a vector,~$p_i$\,, in ${\mathbb R}^d$\,, and an orthogonal transformation,~$A\in {\mathbb R}^{d\times d}\,.$
It follows that matrices
\begin{equation}
\label{eq:Yuri}
\left[\sum_{j=1}^d\sum_{\ell=1}^dc_{ijk\ell}\,\hat p_j\hat p_\ell \right]_{1\leqslant i,k\leqslant d}\in {\mathbb R}^{d\times d}
\end{equation}
and
\begin{equation}
\label{eq:Ivan}
\left[\sum_{j=1}^d\sum_{\ell=1}^d\check c_{ijk\ell}\,p_jp_\ell\right]_{1\leqslant i,k\leqslant d}\in  {\mathbb R}^{d\times d}\,,
\end{equation}
where
\begin{equation*}
\hat{p}_i:=\sum_{j=1}^dA_{i,j}\,p_{j}
\qquad{and}\qquad
\check{c}_{ijk\ell}:=\sum_{m=1}^d\sum_{n=1}^d\sum_{o=1}^d\sum_{q=1}^d
A_{m,i}\,A_{n,j}\,A_{o,k}\,A_{q,\ell}\,c_{mnoq}\,,
\end{equation*}
are similar to one another and, consequently, have the same spectra.
\end{proposition}
\begin{proof}
The fourth-rank tensor,~$c_{ijk\ell}$\,, in ${\mathbb R}^d$ can be viewed as a $d\times d$~matrix, whose entries are $d\times d$~matrices,
\begin{equation*}
C=[C_{ik}]_{1\leqslant i,k\leqslant d}\in\left({\mathbb R}^{d\times d}\right)^{d\times d}\,,	
\end{equation*}
 with $C_{ik}\in{\mathbb R}^{d\times d}$ and $(C_{ik})_{j\ell}:=c_{ijk\ell}$\,.
Thus, matrix~(\ref{eq:Yuri}) can be written as
\begin{equation}
\label{eq:Len}
\left[\,\hat p^t\,C_{ik}\,\hat p\,\right]_{1\leqslant i,k\leqslant d}
=
\left[(A\,p)^t\,C_{ik}\,(A\,p)\right]_{1\leqslant i,k\leqslant d} 
=
\left[\,p^t\left(A^t\,C_{ik}\,A\right)p\,\right]_{1\leqslant i,k\leqslant d} \in{\mathbb R}^{d\times d}\,,
\end{equation}
where ${}^t$ denotes the transpose.

We claim that matrix~(\ref{eq:Ivan}) can be written as
\begin{equation}
\label{eq:Len}
A^t\left[\,p^t\left(A^t\,C_{ik}\,A\right)p\,\right]_{1\leqslant i,k\leqslant d}A\in{\mathbb R}^{d\times d}\,.
\end{equation}
To see this, we let matrix~(\ref{eq:Len}) be
\begin{equation*}
M=A^t\,X\,A\,,
\end{equation*}
where $X:=\left[\,p^t\left(A^t\,C_{ik}\,A\right)p\,\right]_{1\leqslant i,k\leqslant d}$\,, to write 
\begin{equation*}
M_{i,k}
=
\sum_{m=1}^d\sum_{o=1}^d (A^t)_{i,m}\,X_{m,o}\,A_{o,k}
=
\sum_{m=1}^d\sum_{o=1}^d A_{m,i}\,X_{m,o}\,A_{o,k}
\,.
\end{equation*}
Defining $Y:=A^t\,C_{mo}\,A$\,, we have
\begin{equation*}
X_{m,o}=\sum_{j=1}^d\sum_{\ell=1}^d Y_{j,\ell}\,p_{j}\,p_{\ell}\,,
\end{equation*}
where
\begin{equation*}
Y_{j,\ell}
=
\sum_{n=1}^d\sum_{q=1}^d (A^t)_{j,n}\,(C_{mo})_{nq}\,A_{q,\ell}
=
\sum_{n=1}^d\sum_{q=1}^d  A_{n,j}\,A_{q,\ell}\,c_{mnoq}
\,.
\end{equation*}
Hence, 
\begin{equation*}
X_{m,o}
=
\sum_{j=1}^d\sum_{\ell=1}^d\left(
	\sum_{n=1}^d\sum_{q=1}^d  A_{n,j}\,A_{q,\ell}\,c_{mnoq} 
\right)
\end{equation*}
and, in turn,
\begin{equation*}
M_{i,k}
=
\sum_{j=1}^d\sum_{\ell=1}^d \left(\sum_{m=1}^d\sum_{n=1}^d\sum_{o=1}^d\sum_{q=1}^d
A_{m,i}\,A_{n,j}\,A_{o,k}\,A_{q,\ell}\,
c_{mnoq}\right)p_jp_\ell=
\sum_{j=1}^d\sum_{\ell=1}^d \check{c}_{ijk\ell}\,p_jp_\ell\,,\qquad i,k=1\,,\,\ldots\,,\,d\,,
\end{equation*}
which is matrix~\eqref{eq:Ivan}, as required.
\end{proof}
\begin{corollary}
\label{cor:Michael}
From Proposition~\ref{prop:Len}---and the aforementioned fact that the similar matrices share the same spectrum, as well as the fact that the similarity of matrices is not affected by subtracting from them the identity matrices---it follows that
\begin{equation*}
\det\left[\sum_{j=1}^3\sum_{\ell=1}^3c_{ijk\ell}\,\hat p_j\hat p_\ell-\delta_{ik}\right]=\det\left[\sum_{j=1}^3\sum_{\ell=1}^3\check c_{ijk\ell}\,p_jp_\ell-\delta_{ik}\right]\,,\qquad i,k=1,2,3\,,
\end{equation*}
and, hence, equations
\begin{equation}
\label{eq:phat}
\det\left[\sum_{j=1}^3\sum_{\ell=1}^3c_{ijk\ell}\,\hat p_j\hat p_\ell-\delta_{ik}\right]=0\,,\qquad i,k=1,2,3\,,
\end{equation}
and
\begin{equation}
\label{eq:chat}
\det\left[\sum_{j=1}^3\sum_{\ell=1}^3\check c_{ijk\ell}\,p_jp_\ell-\delta_{ik}\right]=0\,,\qquad i,k=1,2,3\,,
\end{equation}
are equivalent to one another.
\end{corollary}
\noindent Corollary~\ref{cor:Michael} is valid even without requiring the index symmetries of Hookean solids.
Also, $A\in O(3)$\,, not only $A\in SO(3)$\,, which is more general than the property invoked by \citet{IvanovStovas2016,IvanovStovas2017}.
\section{Numerical example}
\label{sec:NumEx}
Consider an orthotropic tensor \citep[Table~2]{IvanovStovas2016}, whose components are
\begin{align*}
&c_{1111}=6.3\,,\quad
c_{2222}=6.9\,,\quad
c_{3333}=5.4\,,\\
&c_{1122}=c_{2211}=2.7\,,\quad
c_{1133}=c_{3311}=2.2\,,\quad
c_{2233}=c_{3322}=2.4\,,\\
&c_{1212}=c_{2112}=c_{2121}=c_{1221}=1.5\,,\quad
c_{1313}=c_{3113}=c_{3131}=c_{1331}=0.8\,,\quad
c_{2323}=c_{3223}=c_{3232}=c_{2332}=1.0\,.
\end{align*}
Also, consider vector~$p=\Big[0,0,\sqrt{\tfrac{1}{c_{3333}}}\,\Big]$\,.
Rotating this vector by
\begin{equation}
\label{eq:Rot}
A=
\left[\begin{array}{ccc}
1 & 0 & 0\\
0 & \cos\theta & -\sin\theta\\
0 & \sin\theta & \cos\theta
\end{array}\right]\,,
\end{equation}
with an arbitrary angle of $\theta=\tfrac{\pi}{5}$\,,
and the tensor by $A^t$\,, we obtain
\begin{equation*}
\hat\Gamma
:=
\sum_{j=1}^3\sum_{\ell=1}^3c_{ijk\ell}\,\hat p_j\hat p_\ell
=
\left[\begin{array}{ccc}
0.192934 & 0 & 0\\
0 & 0.331749 & -0.0184241\\
0 & -0.0184241 & 0.949406
\end{array}\right]
\end{equation*}
and
\begin{equation*}
\check\Gamma
:=
\sum_{j=1}^3\sum_{\ell=1}^3\check c_{ijk\ell}\,p_jp_\ell
=
\left[\begin{array}{ccc}
0.192934 & 0 & 0\\
0 & 0.562667 & -0.299407\\
0 & -0.299407 & 0.718488
\end{array}\right]\,,
\end{equation*}
respectively.
The eigenvalues of these matrices are the same, $\lambda_1=0.949955$\,, $\lambda_2=0.33120$ and $\lambda_3=0.192934$\,, as required for similar matrices.
Their corresponding eigenvectors are related by transformation~(\ref{eq:Rot}).

Herein, $\det[\hat\Gamma-I]=-0.027013=\det[\check\Gamma-I]$\,.
In general, the two determinants are equal to one another.
Hence, if $\det[\hat\Gamma-I]=0$\,, so does $\det[\check\Gamma-I]$\,, and {\it vice versa}.

The equivalence of equations~(\ref{eq:phat}) and (\ref{eq:chat}) does not imply their equivalence to
\begin{equation*}
\det\left[\sum_{j=1}^3\sum_{\ell=1}^3c_{ijk\ell}\,p_jp_\ell-\delta_{ik}\right]=0\,,\qquad i,k=1,2,3\,.
\end{equation*}
The eigenvalues of
\begin{equation*}
\Gamma
:=
\sum_{j=1}^3\sum_{\ell=1}^3c_{ijk\ell}\,p_jp_\ell
=
\left[\begin{array}{ccc}
0.148148 & 0 & 0 \\
0 & 0.185185 & 0\\
0 & 0 & 1
\end{array}\right]
\end{equation*}
are $\lambda_1=0.148148$\,, $\lambda_2=0.185185$ and $\lambda_3=1$\,, which are distinct from the eigenvalues of $\hat\Gamma$ and $\check\Gamma$\,.
Herein---in view of $p$ and $c_{ijk\ell}$ representing, respectively, the slowness vector along the~$x_3$-axis and its corresponding elasticity tensor---$\det[\Gamma-I] = 0$\,, which results in the eikonal equations.
We emphasize, however, that Proposition~\ref{prop:Len} and Corollary~\ref{cor:Michael} are valid for arbitrary vectors and fourth-rank tensors, even though, in this example, they are related by the Christoffel equations.
\section*{Acknowledgments}
The authors wish to acknowledge Igor Ravve for bringing this issue to their attention, Yuriy Ivanov for presenting and clarifying aspects of this problem, Sandra Forte for fruitful discussions, and David Dalton for insightful proofreading.

This research was performed in the context of The Geomechanics Project supported by Husky Energy. 
Also, this research was partially supported by the Natural Sciences and Engineering Research Council of Canada, grant~202259.
\bibliography{BSSVarXiv.bib}
\bibliographystyle{apa}
\end{document}